\documentclass[conference]{IEEEtran}

\usepackage{mathtools}
\usepackage{float}
\usepackage{graphicx} 
\graphicspath{ {images/} }

\usepackage[utf8]{inputenc}
\usepackage{amsthm,amssymb}

\usepackage{optidef}

\theoremstyle{definition}
\newtheorem{example}{Example}

\newtheorem{lemma}{Lemma}

\newtheorem{remark}{Remark}

\DeclareMathOperator{\supp}{supp}

\newcommand{\cC}{\mathcal{C}}

\newcommand{\bH}{\mathbf{H}}
\newcommand{\be}{\mathbf{e}}
\newcommand{\bz}{\mathbf{z}}
\newcommand{\bhe}{\hat{\mathbf{e}}}
\newcommand{\he}{\hat{e}}
\newcommand{\br}{\mathbf{r}}
\newcommand{\bh}{\mathbf{h}}
\newcommand{\bs}{\mathbf{s}}

\newcommand{\bc}{\mathbf{c}}

\newcommand{\bv}{\mathbf{v}}

\newcommand{\cH}{\mathcal{H}}
\newcommand{\bzero}{\mathbf{0}}
\newcommand{\bone}{\mathbf{1}}
\newcommand{\ft}{\mathbb{F}_2}

\usepackage{xcolor}
\def\BibTeX{{\rm B\kern-.05em{\sc i\kern-.025em b}\kern-.08em
    T\kern-.1667em\lower.7ex\hbox{E}\kern-.125emX}}

\usepackage{fancyhdr}

\IEEEoverridecommandlockouts                              

\title{
Constrained Error Pattern Generation for GRAND
}
\author{
\IEEEauthorblockN{Mohammad Rowshan, {\em Member, IEEE} and Jinhong Yuan, {\em Fellow, IEEE}}
 \IEEEauthorblockA{School of Electrical Eng. and Telecom., University of New South Wales (UNSW), Sydney, Australia\\
 \{m.rowshan,j.yuan\}@unsw.edu.au}
 \thanks{The work was supported in part by the Australian Research Council (ARC) Discovery Project under Grant DP220103596.}
}

\begin{document}

\maketitle
\pagestyle{empty}
\thispagestyle{fancy}
\lhead{To appear in the proceedings of the 2022 IEEE International Symposium on Information Theory (ISIT 2022)}
\cfoot{}

\begin{abstract}
Maximum-likelihood (ML) decoding can be used to obtain the optimal performance of error correction codes. However, the size of the search space and consequently the decoding complexity grows exponentially, making it impractical to be employed for long codes. In this paper, we propose an approach to constrain the search space for error patterns under a recently introduced near ML decoding scheme called guessing random additive noise decoding (GRAND). In this approach, the syndrome-based constraints which divide the search space into disjoint sets are progressively evaluated. By employing $p$ constraints extracted from the parity check matrix, the average number of queries reduces by a factor of $2^p$ while the error correction performance remains intact.

\end{abstract}



\section{INTRODUCTION}
\label{sec:intro}
Soft-decision based decoding can be classified into two major categories of decoding algorithms \cite{lin}: Code structure-based algorithms and reliability-based algorithms or general decoding algorithms as they usually do not depend on the code structure. In the reliability-based algorithms, which is the focus of this paper, the goal is to find the closest modulated codeword to the received sequence using a metric such as likelihood function. That is, we try to maximize the likelihood in the search towards finding the transmitted sequence. Hence, this category of decoding algorithms is called maximum likelihood (ML) decoding which is known as an optimal decoding approach. Maximum likelihood decoding has been an attractive subject for decades among the researchers. Error sequence generation is one of the central problems in any ML decoding scheme. 
In general, a ML decoding is prohibitively complex for most codes as it was shown to be an NP-complete problem \cite{berlekamp}. Hence, the main effort of the researchers has been  concentrated on reducing the algorithm's complexity for short block-lengths. The complexity reduction has inevitably come at the cost of error correction performance degradation. 

Information set decoding \cite{dorsch} and its variant, ordered statistics decoding (OSD) \cite{fossorier}, 
are a well-known examples of near-ML decoding algorithms. The OSD algorithm permutes the columns of the generator matrix with respect to the reliability of the symbols for every received vector and performs elementary row operations on the independent columns extracted from the permuted generator matrix resulting in the systematic form.  
The recently introduced \emph{guessing random additive noise decoding} (GRAND) \cite{duffy-tit} suggests generating the error sequences in descending order of likelihood that does not require any matrix manipulation for every received vector. 
The likelihood ordering of error sequences is determined accurately in soft-GRAND \cite{duffy-sgrand} and approximately  in \emph{ordered reliability bits} GRAND (ORBGRAND) \cite{duffy-orbgrand,abbas-orbgrand,riaz}. The accurate ordering requires sorting a large number of metrics while the approximate scheduling of the error sequences is based on distinct integer partitioning of positive integers which is significantly less complex. 

In this work, by utilizing the structure of a binary linear code, 
we propose an efficient pre-processing that constrains the error sequence generation. This approach can save the codebook checking operation which seems to be more computational than the pre-processing for validation. 
These syndrome-based constraints are extracted from the parity check matrix (with or without matrix manipulation) of the underlying code. The proposed approach breaks down the unconstrained search problem 
into the problem of dividing the scope of search into disjoint set(s) and determining the number of errors (even or odd) in each set. We show that the size of the search space deterministically reduces by a factor of $2^p$ where $p$ is the number of constraints. The numerical results shows a similar reduction in the average number of queries.  Note that the constrained error sequence generation does not degrade the error correction performance as it is just discarding the error sequences that do not result in valid codewords. The proposed approach can be applied on other GRAND variants such as soft-GRAND \cite{duffy-sgrand} and GRAND-MO \cite{an-grand-mo}.

\section{PRELIMINARIES}\label{sec:prelim}
We denote by $\ft$ the binary finite field with two elements. The cardinality of a set is denoted by $|\cdot|$.  The interval $[a,b]$ represents the set of all integer numbers in $\{x:a\leq x\leq b\}$. The \emph{support} of a vector $\be = (e_1,\ldots,e_{n}) \in \ft^n$ is the set of indices where $\be$ has a nonzero coordinate, i.e. $\supp(\be) \triangleq \{i \in [1,n] \colon e_i \neq 0\}$ . 
The \emph{weight} of a vector $\be \in \ft^n$ is $w(\be)\triangleq |\supp(\be)|$. The all-one vector $\bone$ and all-zero vector~ $\bzero$ are defined as vectors with all identical elements of  0 or 1, respectively. The summation in $\ft$ is denoted by $\oplus$.
\subsection{ML Decoding and ORBGRAND}
A binary code $\mathcal{C}$ of length $n$ and dimension $k$ maps a message of $k$ bits into a codeword $\mathbf{c}$ of $n$ bits to be transmitted over a noisy channel. The channel alters the transmitted codeword such that the receiver obtains an $n$-symbol vector~$\mathbf{r}$. A ML decoder supposedly compares $\mathbf{r}$ with all the $2^k$ modulated codewords in the codebook, and selects the one closest to $\mathbf{r}$. In other words, the ML decoder for Gaussian noise channel finds a modulated codeword $\textup{x}(\bc)$ such that
\begin{equation}
    \hat{\mathbf{c}} = \underset{\mathbf{c}\in\mathcal{C}}{\text{arg max }} p(\mathbf{r}|\textup{x}(\mathbf{c}))=\underset{\mathbf{c}\in\mathcal{C}}{\text{arg min }}|\mathbf{r}-\textup{x}(\mathbf{c})|^2.
\end{equation}
This process requires checking possibly a large number of binary error sequences $\bhe$ to select the one that satisfies 
\begin{equation}\label{eq:check}
     \mathbf{H} \cdot (\theta(\mathbf{r})\oplus\hat{\mathbf{e}}) = \mathbf{0}
\end{equation}
where $\bH$ is the parity check matrix of code $\cC$ and $\theta(\mathbf{r})$ returns the hard-decision demodulation of the received vector $\mathbf{r}$. 
The error sequence $\bhe$ in ORBGRAND is generated in an order suggested by {\em logistic weight} $w_L$ which is defined as \cite{duffy-orbgrand} 
\begin{equation}\label{eq:w_L}
    w_L(\bz)=\sum_{i=1}^n z_i\cdot i
\end{equation}
where $z_i\in\ft$ is the $i$-th element of the permuted error sequence $\hat{\mathbf{e}}$ in the order of the received symbols's reliability which is proportional to their magnitudes $|r_i|$. That is, the error sequence is  $\bhe=\pi_1(\mathbf{z})$ where $\pi_1(.)$ is the permutation function. 
The {\em integer partitions} of positive integers from 1 to $n(n+1)/2$ with distinct parts and no parts being larger than $n$ gives the indices of the positions $i$ in the permuted error sequences $\mathbf{z}$ to be $z_i=1$. 
Alternatively, a sequential method to generate error patterns was suggested based on partial ordering and a modified logistic weight  in \cite{condo-grand}.

\subsection{Parity Check Matrix}
The matrix $\bH$ represents the properties of the underlying code. These properties can be exploited to constrain the process of generating $\bz$ in \eqref{eq:w_L} and consequently constraining $\bhe$. 
We know that the syndrome $\mathbf{s}$ for any vector $\bv$ of length $n$ is obtained by $\mathbf{s} = \mathbf{H}\cdot\bv$ 
and a valid codeword $\bc$ gives $\bH\cdot\bc=\bzero$, where the syndrome is $\bs=\bzero$. Mathematically speaking, all $\bc^\perp\in\cC$ form the null space for $\bH$. Let us represent the parity check matrix as follows:
\begin{equation}
    \bH = [
        \bh_1\,
        \bh_2\,
        \cdots\,
        \bh_{n-k}
    ]^T
\end{equation}
The $n$-element vectors $\bh_j$ for $j\in[1,n-k]$ are the rows of the parity check matrix denoted by $\bh_j = [h_{j,1} \; h_{j,2}  \;\cdots \;h_{j,n}]$. If we define vector $\bv=\theta(\br)$, then we are looking for some error sequence $\bhe=\pi_1(\bz)$ with elements $\he_i\in\ft,i=1,...,n$ such that $\mathbf{H}\cdot(\bv\oplus\bhe)=\bzero$. When $\he_i=1$, it means the $i$-th received symbol is supposed to be in error. 

\subsection{Manipulation of the Parity Check Matrix}
Now we turn our focus on the manipulation of the parity check matrix so that we can use it to divide the scope of search into disjoint sets of positions. 
We know that the rows $\bh_j,j=[1,n-k]$ of the parity check matrix $\bH$ of the code $\cC$ form a basis for the dual code $\cC^\perp$. That is, all $\bc^\perp\in\cC^\perp$ form the row space of $\bH$ which are obtained from the linear combinations of $\bh_1,\bh_2, ..., \bh_{n-k}$ in $\ft$. 
From the linear algebra, we know that the row space is not affected by elementary row operations on $\bH$ (resulting in $\bH'$), because the new system of linear equations represented in the matrix form $\bH'\cdot\bc=\bzero$ will have an unchanged solution set $\cC$. 
Similarly, permutation of the rows $\bh_j$ does not have any impact on the row space of $\bH$ as well.

\section{Constraining the Error Sequence Generation}

In this section, we first reformulate the syndrome computation, then we propose an approach to utilize the information we obtain from the syndrome $\bs(\bzero)$ for pruning the search space. 
\subsection{Relative Syndrome}
Since $\bv=[v_{1} \; v_{2} \;\cdots \;v_{n}]=\theta(\br)$ remains unchanged for all the queries performed on every received vector $\br$ while the error sequence $\bhe$ changes for every query, then the syndrome $\bs$ is a function of $\bhe$. Given vector $\bv\oplus\bhe$ is used for obtaining the syndrome $\bs=[s_{1} \; s_{2} \;\cdots \;s_{n-k}]$, we have 
\begin{equation}\label{eq:synd_elements}
    s_j = \bigoplus_{i=1}^{n} h_{j,i} \cdot (v_i\oplus\he_i)
\end{equation}
for every $j\in[1,n-k]$.  Hence, we can relate every syndrome $\bs$ to the corresponding error sequence $\bhe$ and denote it by $\bs(\bhe)$. The computation of the syndrome of every $\bhe\neq\bzero$ relative to the syndrome of all-zero error sequence $\bhe=\bzero$, i.e., $\bs(\bzero)$, can be simplified as 
\begin{equation}\label{eq:rel_synd}
    s_j(\bhe) = \big(\bigoplus_{i\in\supp(\bhe)} h_{j,i}\big) \oplus s_j(\bzero)
\end{equation}
where the elements of $\bs(\bzero)$ are computed by 
\begin{equation}\label{eq:synd_zero}
    s_j(\bzero) = \bigoplus_{i=1}^{n} \big(h_{j,i} \cdot v_i\big)
\end{equation}
in the first query where $\bv=\theta(\mathbf{r})$ in the presence of all-zero error sequence. Note that \eqref{eq:rel_synd} follows from \eqref{eq:synd_elements} through splitting it into two terms by the distributive law.
\begin{remark}
In order to get $\bs(\bhe)=\bzero$ which is the goal of the search, according to  \eqref{eq:rel_synd}, we need to find a vector $\bhe$ such that
    \begin{equation}\label{eq:e_condition}
    \bigoplus_{i\in\supp(\bhe)} h_{j,i}=|\supp(\bh_j)\cap\supp(\bhe)|\text{ mod }2= s_j(\bzero)\\ 
    \end{equation}
    for every $j\in[1,n-k]$.
\end{remark}

\subsection{The Location and Number of Errors}
To utilize \eqref{eq:e_condition} for pruning the search space, we form one or more disjoint sets of bit indices, denoted by $\cH_j,j\in[1,n-k]$ where 
$\cH_j=\supp(\bh_j)$. 
Note that $\bh_j$ could belong to the original parity check matrix $\bH$ or the manipulated one. We discuss about how to manipulate $\bH$ effectively later.

\begin{remark}\label{rmk:even_odd}
    Observe that a general insight into the number of erroneous positions in set $\cH_j=\supp(\bh_j)$ is given by
    \begin{equation}\label{eq:even_odd}
        |\cH_j\cap\supp(\be)|=
        \begin{dcases}
            \text{odd}  & s_j(\bzero)=1\\ 
            \text{even} & s_j(\bzero)=0
        \end{dcases}
    \end{equation}
    where the even number of errors includes no errors as well. 
    When $|\cH_j|\!=\!n$ or $\bh_j\!=\!\bone$, we have $|\cH_j\cap\supp(\be)|\!=\!|\supp(\be)|$. Then, the whole sequence $\be$ has either odd or even weight.  
\end{remark}


After establishing the required notations and reformulating the background, we come to the main idea. We use the structure of matrix $\bH$ to form sets $\cH_j$. These sets help us to break down the problem of finding the location of errors among $n$ bit positions into the problem of finding the number of errors (even or odd) in intervals. 
We consider several scenarios in the following:

The simplest case is as follows: In many codes such as polar codes, PAC codes, and extended codes by one bit, there exists a row $\bh_j$ such that $h_{j,i}=1$ for every $i\in[1,n]$. The parity corresponding to such a row is called {\em overall parity check} bit. In this case, we have $\cH_j=[1,n]$. Following \eqref{eq:even_odd}, we would know when we have even or odd number of errors in total. Then, we can use this relation to discard the irrelevant error sequences easily in the stage of generation. 

Let us make it a bit more complicated: Pick a row $\bh_j$ of $\bH$ and its corresponding set $\cH_j=\supp(\bh_j)$ such that $\cH_j\subset[1,n]$. In this case, we do not know whether the total number of errors is even or odd but we can obtain this insight on a specific set of bit positions, i.e., $\cH_j$. The number of erroneous bits outside this set could be either even or odd. Implementation of $|\cH_j\cap\,\supp(\be)|$ could be challenging. An easy way to compute this is via a permutation $\pi_2(.)$ that converts the scattered indices in each set into consecutive indices and obtaining an interval $[L_j,U_j]$, similar to \eqref{eq:permutation_2}. Now, by counting every $i\in\supp(\bhe)$ such that $L_j\leq\pi_2(i)\leq~U_j$ or simply $\pi_2(i)>L_j-1$ for a single constraint, we can find 
\begin{equation}
   |\cH_j\cap\supp(\bhe)| = |\{i\in\supp(\bhe):\pi_2(i)\in[L_j,U_j]\}|.
\end{equation}
Then, the generated $\bhe$ can be output if the condition in \eqref{eq:e_condition} is satisfied. 
Note that a more efficient and fast way to the the aforementioned process can be implemented based on the structure of matrix $\bH$ and a good choice of $\bh_j$. For instance, the elements of $\bh_j$ in polar codes' matrix $\bH$ are already in the consecutive order. 


Now, we turn our attention to multiple sets of $\cH_j$. These sets are preferred to be disjoint, then every element in the error sequence will belong to only one interval. This makes the checking simple. An effective approach to get such sets is as follows: Pick a row $\bh_{j_1}$ with the largest $|\supp(\bh_{j_1})|$ and then try to find a row $\bh_{j_2}$ such that $\supp(\bh_{j_2}) \subset \supp(\bh_{j_1})$. In this case, we can replace $\bh_{j_1}$ with $\bh_{j} = \bh_{j_1}\oplus\bh_{j_2}$ in order to have two disjoint sets $\cH_{j}$ and $\cH_{j_2}$ for constraining the search space based on \eqref{eq:e_condition}. As we will see in Section \ref{sec:analysis}, every set $\cH_j$ used in constraining the search space reduces the size of the search space by half. This strategy can make decoding of longer codes practical. Fig. \ref{fig:split_search} illustrates how matrix manipulation can split the search space. 

\begin{figure}[ht]
    \centering
    \includegraphics[width=0.7\columnwidth]{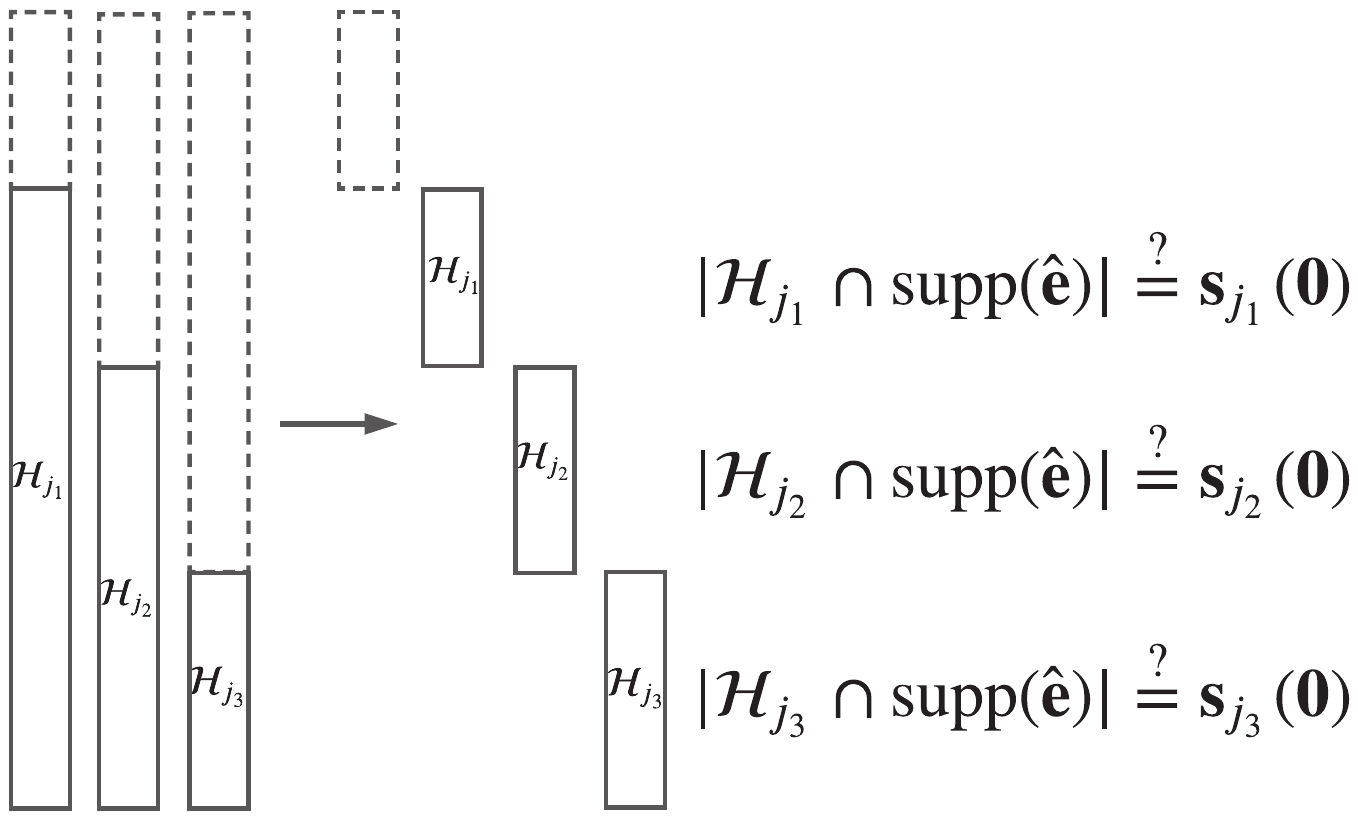}
    \caption{Splitting the permuted indices into three sets that can be represented by three intervals: A sketch showing an example for $\cH_{j_1} = \supp(\bh_{j_1})$, $\cH_{j_2} = \supp(\bh_{j_2})$, $\cH_{j_3} = \supp(\bh_{j_3})$ where $\cH_{j_3} \subset \cH_{j_2} \subset \cH_{j_1}$. The right hand side is obtained from $\bh_{j_1}=\bh_{j_1}\oplus\bh_{j_2}$ and $\bh_{j_2}=\bh_{j_2}\oplus\bh_{j_3}$.}        \label{fig:split_search}
    \vspace{-10pt}
\end{figure}



\begin{example}\label{ex:h1_subset_h2}
Suppose we have two rows of a parity check matrix and the associated syndrome bits as follows:
$$\bh_{j_1} = [1 \; 1  \; 1  \; 1 \; 0  \; 1  \;1  \;0], \;\;s_{j_1}(\bzero)=1$$
$$\bh_{j_2} = [0 \; 1  \; 0  \; 1 \; 0  \; 0  \;1  \;0], \;\;s_{j_2}(\bzero)=0$$
One can conclude from $s_{j_1}(\bzero)=1$ that the error(s) could have occurred on any odd number of positions in set $\{1,2,3,4,6,7\}$, i.e., $$\supp(\bhe) \subset \{1,2,3,4,6,7\}: |\supp(\bhe)|\text{ mod } 2=1,$$ or any even number of positions in set $\{2,4,7\}$, however we do not have any information about positions 5 and 8 based on these two constraints. 
Since $\supp(\bh_{j_2}) \subset \supp(\bh_{j_1})$, then we can replace $\bh_{j_1}$ with
$$\bh_{j}=\bh_{j_1}\oplus\bh_{j_2} = [1 \; 0  \; 1  \; 0 \; 0  \; 1  \;0 \;0]$$
Now, we can search among the error sequences that satisfy the following two conditions, given $s_j(\bzero)=1$: 
$$|\{1,3,6\}\cap\supp(\bhe)| \text{ mod } 2 = 1,$$
$$|\{2,4,7\}\cap\supp(\bhe)| \text{ mod } 2 = 0.$$
To make the intersection operation easier, we can employ the following permutation:
\begin{equation}\label{eq:permutation_2}
    \pi_2:\{5,8,{\color{blue}1,3,6},{\color{red}2,4,7}\}\rightarrow\{1,2,{\color{blue}3,4,5},{\color{red}6,7,8}\}.
\end{equation}
Now, to count all $i\in\supp(\bhe)$ in $\{2,4,7\}\cap\supp(\bhe)$ or in $\{1,3,6\}\cap\supp(\bhe)$, we just need to count every $i\in\supp(\bhe)$ where  $\pi_2(i)>5$ or $\pi_2(i)>2$. Note that $i$ cannot satisfy both, then if $\pi_2(i)>5$ is true, we do not check the other. 
\end{example}

\begin{remark}\label{rmk:pass_seq}
    Observe that an error sequence corresponding to $\bz$ under the single constraint $\bh_j$ is valid if
    \begin{equation}\label{eq:pass_seq2}
        \bs_j(\bhe)=\bs_j(\pi_1(\bz))=\bigoplus_{\substack{i\in\supp(\bz):\\ L_j\leq \pi_2(\pi_1(i))\leq U_j}} 1 = \bs_j(\bzero).
    \end{equation}
    The summation in \eqref{eq:pass_seq2} can be performed progressively as the sequences $\supp(\bz)$ differ in a limited number of elements between the parent and children sequences during generation.
\end{remark}    

\subsection{Progressive Evaluation of the Constraints}
For further simplification of embedding the constraints into the error sequence generator, we can keep the partial evaluation of the left hand side of  \eqref{eq:pass_seq2} for every sequence during the integer partitioning process. 
Suppose we have $\bz$ generated during integer partitioning. Then, the partial evaluation of the constraint which excludes the largest element of $\supp(\bz)$ is defined as
\begin{equation}
    \bs_j^*(\bhe)=\bs_j^*(\pi_1(\bz))=\bigoplus_{\substack{i\in\supp(\bz^*):\\ L_j\leq \pi_2(\pi_1(i))\leq U_j}} 1,
\end{equation}
where $\supp(\bz^*)=\supp(\bz)\backslash\{\max(\supp(\bz))\}$. This partial evaluation can be used to progressively compute $\bs_j(\bhe)$ for longer $\bhe$ sequences. This can be appreciated when we recognize two main phases in the generation of the error sequences: 1) partitioning the largest element of the current sequence into two integers, 2) finding the alternative pair of integers for the result of the phase 1. These two phases are repeated until we can no longer do the partitioning in phase 1. Fig. \ref{fig:partitionig} illustrates an example showing these two phases. For every $\bz$ corresponding to an error sequence, we keep a partial evaluation of the constraint such that we do not involve the following elements in \eqref{eq:rel_synd}: 1) The largest  element subject to partitioning in the first phase, or 2) the two largest elements that we are seeking their alternatives in the second phase. This way, the computation of $\bs_j(\bhe)$ becomes simple.

To evaluate $\bs_j(\bhe)$ for every sequence $\bhe$, the two new integers during the two phases, denoted by $i_{\ell+1}$ and $i_{\ell+2}$ where $\ell=|\supp(\bz^*)|$, are checked to be in $[L_j,U_j]$ as follows: 
\begin{equation}\footnotesize
    \bs_j(\bhe)= 
        \begin{dcases}
            \bs_j^*(\bhe)  & \pi_2(\pi_1(i_{\ell+1}))\text{ and }\pi_2(\pi_1(i_{\ell+2}))\in[L_j,U_j],\\ 
            \bs_j^*(\bhe)\oplus 1  & \pi_2(\pi_1(i_{\ell+1}))\text{ or }\pi_2(\pi_1(i_{\ell+2}))\in[L_j,U_j],\\ 
            \bs_j^*(\bhe) & \text{otherwise.}
        \end{dcases}
\end{equation}
\begin{figure}[ht]
    \centering
    \includegraphics[width=0.9\columnwidth]{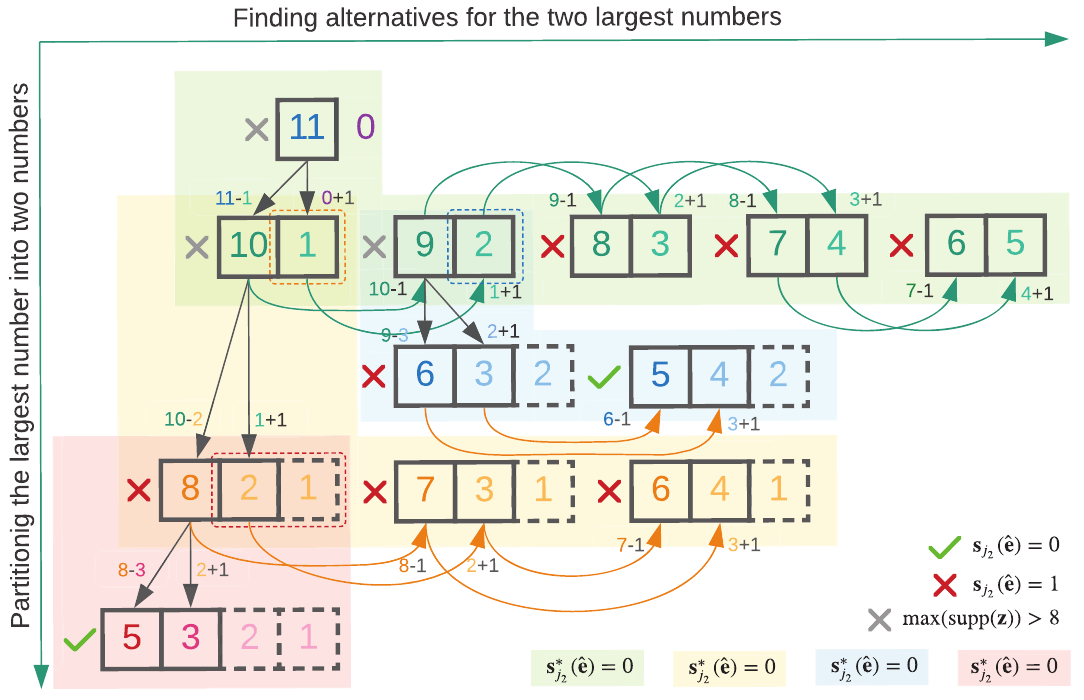}
    \caption{Integer partitioning of $w_L=11$ (Example \ref{ex:partitions_w_L=11}). Red crosses and gray crosses indicate the invalidated sequences by constraint $\bh_{j_2}$ and constraint $i>n$ for $n=8$, respectively. The sequences with identical background color use the same $\bs^*_{j_2}(\bhe)$ for computing $\bs_{j_2}(\bhe)$ during partitioning the largest number (vertical arrows) and finding the alternatives (horizontal curved arrows). }    \label{fig:int_part}
    \label{fig:partitionig}
    \vspace{-10pt}
\end{figure}

\begin{example}\label{ex:partitions_w_L=11}
    Suppose for the code in Example \ref{ex:h1_subset_h2}, the permutation for the order of symbol reliability is $$\pi_1:\{1,2,3,4,5,6,7,8\}\rightarrow\{8,1,5,6,3,7,2,4\}$$ and $w_L=11$. We use $\bh_{j_2}$ with $\bs_{j_2}=0$ as a single constraint on the error sequence generation. As shown in Fig. \ref{fig:partitionig}, $\{11\}$, $\{10,1\}$, and $\{9,2\}$ are not valid as there exists $i>8$ in the sets. For the reset of the sequences,  the condition in Remark \ref{rmk:pass_seq} is not satisfied (indicated with red crosses) except for two sequences checked with green ticks.
\end{example}

\section{Analysis of Reduction in Size of Search Space}\label{sec:analysis}
The size of the search space denoted by $\Omega$ is theoretically $\Omega=\sum_{\ell\in[0,n]} {n \choose \ell}=2^n$ for binary symbols if we do not put any constraints in place. However, we can find a valid codeword with a limited number of queries when we search by checking the sequences in the descending order of their likelihood. 
Nevertheless, the average number of queries (computational complexity) of decoding reduces proportional to the reduction in the search space given the maximum number of queries, denoted by $b$, is large enough.  Therefore, we focus on the analysis of the reduction in the search space in the presence of the constraint(s) as an equivalent analysis on the reduction in the average complexity.  

To obtain the size of the constrained search space, we consider different single constraints as follows:

When there exists a row $\bh_j=\bone$ in $\bH$: 
    In this case, we only consider the error sequences $\bhe$ where $|\supp(\bhe)| \text{ mod } 2 = s_j(\bzero)$. 
    That is, $|\supp(\bhe)|$ is either odd (when $s_j(\bzero)=1$) or even (when $s_j(\bzero)=0$), not both. Then, the size of the search space will be 
    \begin{equation}
        \Omega(\bh_j)=\sum_{\ell\in[0,n]:\\\ell\text{ mod } 2 = s_j(\bzero)} {n \choose \ell} = \frac{2^n}{2}=2^{n-1}.
    \end{equation}

Consider a row $\bh_j\neq\bone$ and $\cH=\supp(\bh_j)$: 
    In this case, we only consider the error sequences satisfying $|\cH\cap\supp(\bhe)|\text{ mod }2= s_j(\bzero)$. Then, the size of the constrained search space will be 
    \begin{equation}\label{eq:single_const}
        \Omega(\bh_j) = \sum_{\footnotesize\substack{\ell\in[0,|\cH|]:\\ \ell\text{ mod } 2 = s_j(\bzero)}} {|\cH| \choose \ell} \cdot 2^{n-|\cH|} = \frac{2^{|\cH|}}{2} \cdot 2^{n-|\cH|} = 2^{n-1}.
    \end{equation}
As can be seen, the search space halves in both scenarios. However, implementation of the first scenario in the error sequence generator is quite simple with negligible computation overhead.
The following lemma generalizes the size of constrained search space by the number of constraints.

\begin{lemma}\label{lma:seach_space_size}
Suppose we have a parity check matrix $\bH$ in which there are $p$ rows of $\bh_{j},j=j_1,j_2,...,j_p$ with mutually disjoint index sets $\cH_j = \supp(\bh_{j})$, then the size of the constrained search space by these $p$ parity check equations is
    \begin{equation}\label{eq:search_space_size_less_than}
        \Omega(\bh_{j_1},..,\bh_{j_p})  = 2^{n-p}
    \end{equation}
\end{lemma}
\begin{proof}
Generalizing \eqref{eq:single_const} for $p$ constraints, we have
    \begin{multline*}
        \Big(\prod_{j=j_1}^{j_p}\sum_{\footnotesize\substack{\ell\in[0,|\cH_j|]:\\\ell\text{ mod } 2 = s_{j}(\bzero)}} {|\cH_{j}| \choose \ell}\Big) \cdot 2^{n-\sum_{j=j_1}^{j_p}|\cH_{j}|} =\\ \big(\prod_{j=j_1}^{j_p} 2^{|\cH_j|-1}\big) \cdot 2^{n-\sum_{j=j_1}^{j_p}|\cH_{j}|} = 2^{n-p}.
    \end{multline*}\vspace{-10pt}
\end{proof}
Note that the error sequences are not evenly distributed in the search space with respect to the constraint(s). Nevertheless, when the maximum number of queries $b$ increases, by the law of large numbers, the reduction in the average queries approaches the reduction in the size of the search space.

\begin{figure}
    \centering
    \includegraphics[width=0.7\columnwidth]{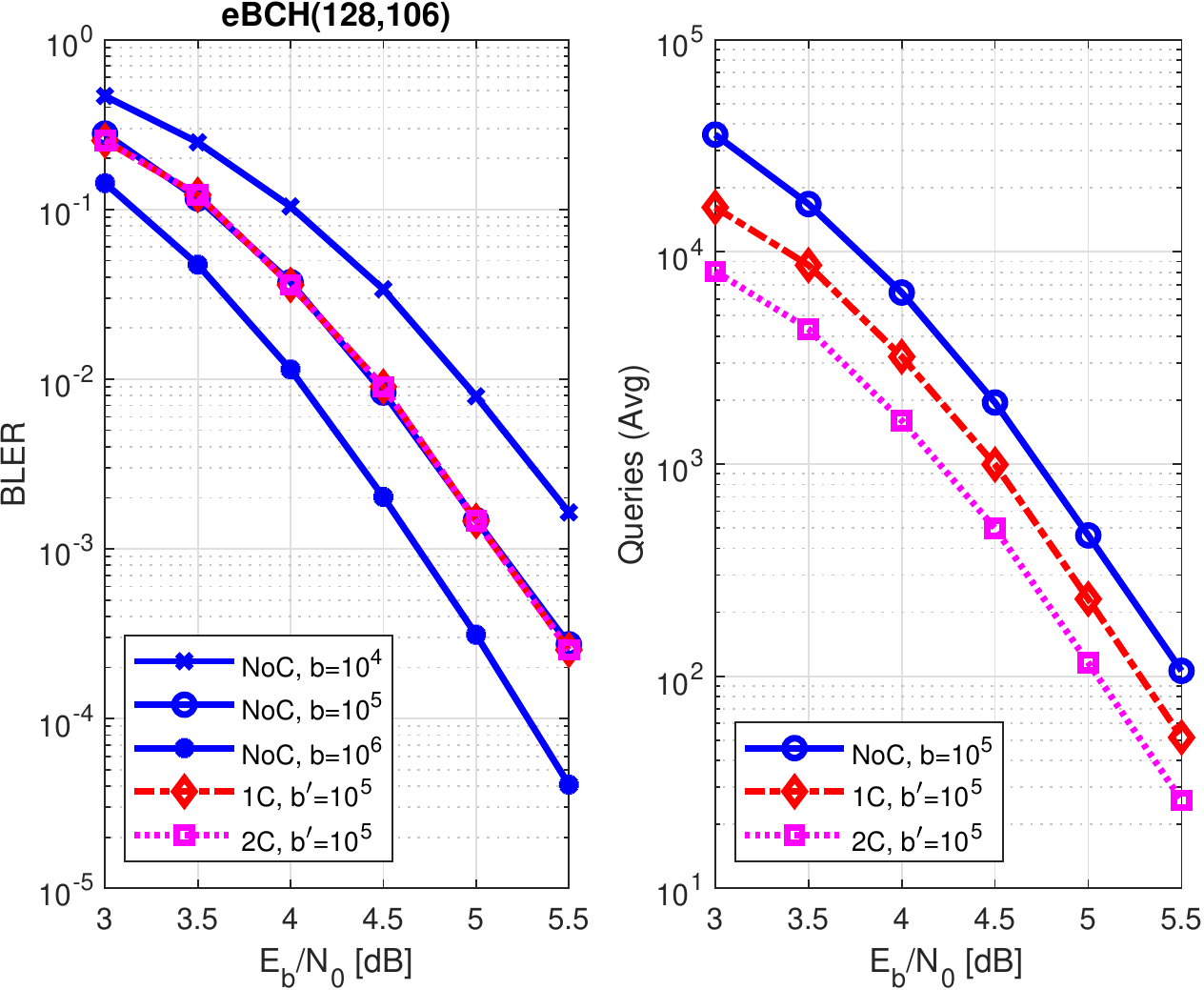}
    \caption{Performance and average queries for eBCH(128,106) under no constraint (NoC), one constraint (1C:$\bh_1)$, and two constraints (2C:$\bh_1,\bh_2)$.}
    \label{fig:bler_ebch}
    \vspace{-10pt}
\end{figure}
\begin{figure}
    \centering
    \includegraphics[width=0.7\columnwidth]{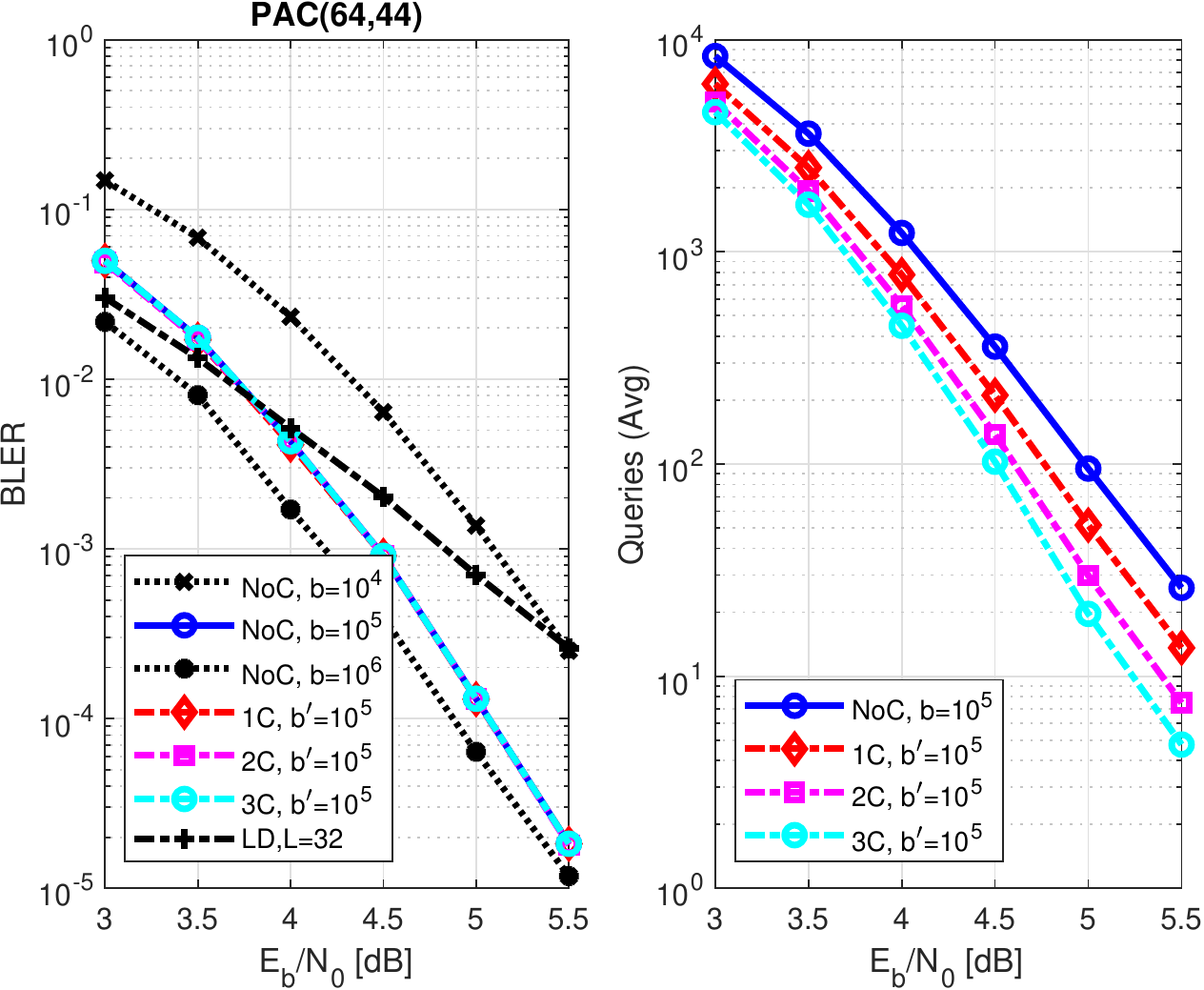}
    \caption{Performance and average queries for PAC(64,44) with precoding polynomial $[1\,0\,1\,1\,0\,1\,1]$ under no constraint (NoC), one constraint (1C:$\bh_1)$, two constraints (2C:$\bh_1,\bh_4)$, and three constraints (3C:$\bh_1,\bh_4,,\bh_5)$. The curve "LD, L=32" shows the BLER under List Decoding \cite{rowshan-pac1} with list size~32.}
    \label{fig:bler_pac}
    \vspace{-10pt}
\end{figure}

\thispagestyle{fancy}
\lhead{You may find the Python implementation of the algorithms used in this paper in \cite{code}.}
\cfoot{}

\section{NUMERICAL RESULTS} \label{sec:results}
We consider two sample codes for numerical evaluation of the proposed approach. Fig. \ref{fig:bler_ebch} and \ref{fig:bler_pac} show the block error rates (BLER) of extended BCH code (128, 106) and polarization-adjusted convolutional (PAC) \cite{arikan2} code (64, 44), respectively. 
Note that $b$ is the maximum number of queries on \eqref{eq:check} while $b'$ is the maximum number of considered error sequences (including discarded ones) in the sequence generator. Observe that when no constraints are applied, then $b=b'$.  For the sake of fair comparison, we take an identical maximum number of considered error sequences. That is, $b'$ for the constrained case equals $b$ of no constraints case. Fig. \ref{fig:success_query} illustrates how constraining can reduce the complexity. 
Fig. \ref{fig:bler_ebch} and \ref{fig:bler_pac} illustrate that the average number of queries halves by every introduced constraint for $b=b'=10^5$ in the presence of different number of constraints ($p=1,2,3$, indicated by $p$C) while the BLERs remain unchanged. When comparing the complexity curves, note  that the vertical axis is in logarithm-scale. The black dotted curves  on BLER figures for $b=10^4$ and $10^6$ are only to show the relative performance with respect to $b$. In the table below, we show the average queries of three cases (without constraints, constrained by either $\bh_1$ or $\bh_2$, and by both $\bh_1$ and $\bh_2$) for the maximum queries of $b=10^5$ at different SNRs. Although we do not compare the results with a CRC-polar code with a short CRC or an improved PAC code (see \cite{rowshan-err_coef,rowshan-precoding}), they can perform close 
in high SNR regime. 
\begin{center}
\footnotesize
\begin{tabular}{|c||c|c|c|c|c|c|} 
 \hline
$E_b/N_0$ [dB] & 3 & 3.5 & 4 & 4.5 & 5 & 5.5 \\ 
 \hline\hline
No Constraints  & 35686 & 16838 & 6430 & 1949 & 461 & 106 \\ 
 \hline
Single: $\bh_1 \text{ or } \bh_2$ & 16183 & 8654 & 3205 & 994 & 231 & 51 \\ 
 \hline
Double: $\bh_1 \text{ and } \bh_2$ & 8091 & 4327 & 1602 & 497 & 115 & 26 \\ 
 \hline
\end{tabular}
\end{center}
As can be seen, the average queries is reduced by a factor of two after adding every constraint. At lower maximum queries, $b<10^5$, we observe slightly less reduction. For instance, when $b=10^4$, the average queries at $E_b/N_0=5$ for these three constraint cases are 205, 144, and 102, respectively.

Note that the rows $\bh_{1}$ and $\bh_{2}$ in $\bH$ matrix for eBCH code (128, 106) satisfy the relationship $\supp(\bh_{2})\subset\supp(\bh_{1})$ where $\bh_2=\bone$ and $|\bh_1|/2=|\bh_2|=64$. Hence, for two constraints, we modify $\bh_1$ by $\bh_1=\bh_1\oplus\bh_2$. Similarly, the rows $\bh_{1}$, $\bh_{4}$, and $\bh_{5}$ in $\bH$ matrix for PAC code (64, 44) satisfy the relationship $\supp(\bh_{5})\subset\supp(\bh_{4})\subset\supp(\bh_{1})$ where $\bh_1=\bone$ and $|\bh_1|/2=|\bh_4|=2|\bh_5|=32$.


\section{CONCLUSION} 
In this paper, 
we propose an approach to constrain the error sequence generator by a single or multiple constraints extracted from the parity check matrix (with or without matrix manipulation) of the underlying codes. We further present how to progressively evaluate the constraints for the error sequences during the generation process. The theoretical analysis shows that the search space reduces by factor of 2 after introducing every constraint. The numerical results support this theory when the maximum number of queries is relatively large. To have a negligible computational overhead, it is suggested to use only one constraint preferably with codes which have an overall parity check bit such as polar codes and their variants and extended codes by one bit. This single constraint suffices to halve the complexity. The proposed method can be applied on other GRAND variants as well.

\begin{figure}
    \centering
    \includegraphics[width=0.65\columnwidth]{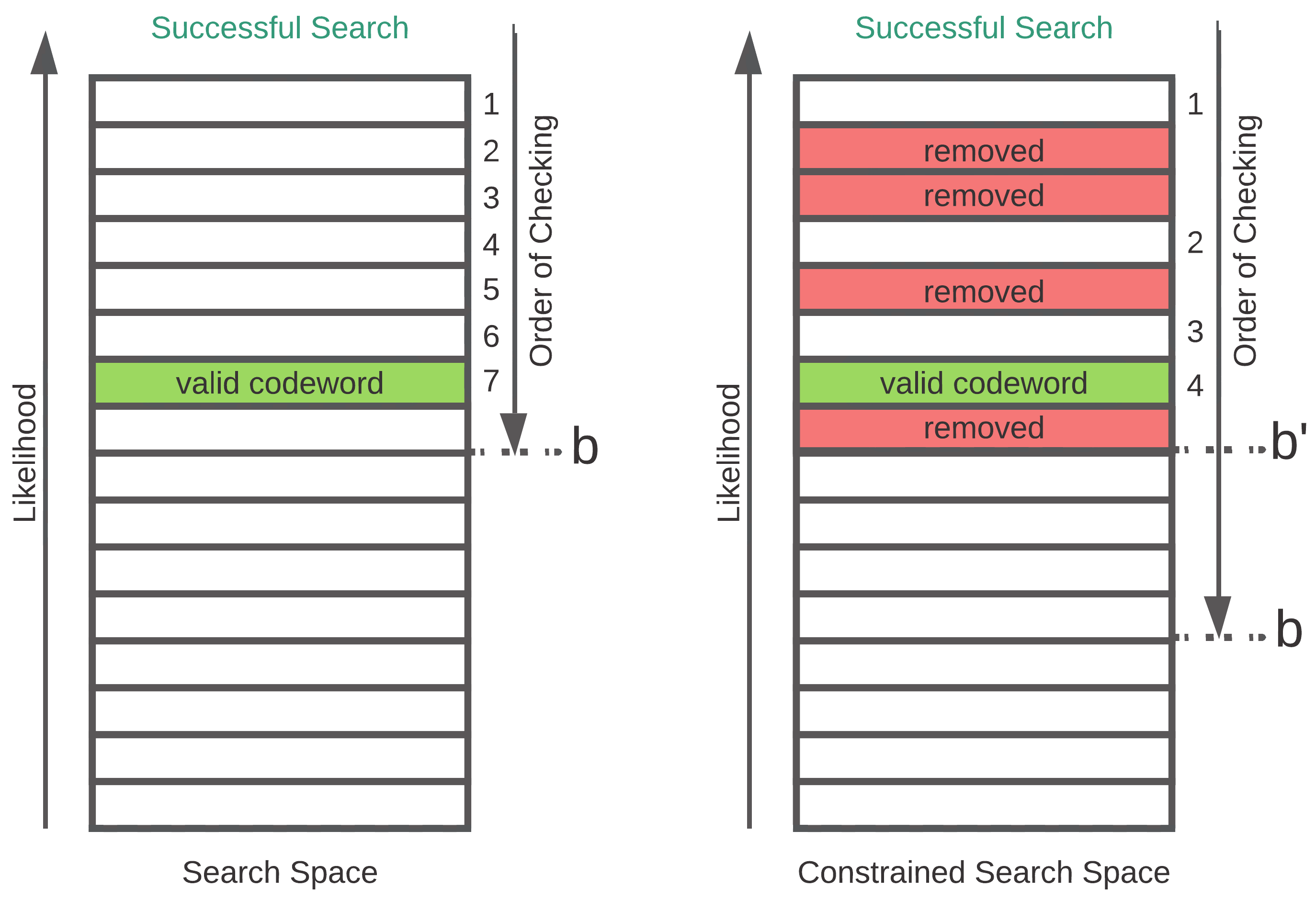}
    \caption{Reducing complexity : An example showing a stack of candidate codewords sorted with respect to likelihood (the codeword at the top has the highest likelihood). By removing the invalid codewords, we can reach to the first valid codeword faster before reaching the abandonment threshold $b=8$.}
    \label{fig:success_query}
    \vspace{-10pt}
\end{figure}

\addtolength{\textheight}{-12cm}   







\newpage


\begin{thebibliography}{9}
\bibitem{lin} S. Lin and D. J. Costello, ``Error Control Coding," 2nd Edition, Pearson Prentice Hall, Upper Saddle River, 2004, pp. 544-550.
\bibitem{berlekamp} E. Berlekamp, R. McEliece, and H. Van Tilborg, ``On the inherent intractability of certain coding problems (corresp.)," {\em IEEE Tran. Inf. Theory}, vol. 24, no. 3, pp. 384-386, 1978.
\bibitem{dorsch} B. Dorsch, ``A decoding algorithm for binary block codes and J-ary output channels," {\em IEEE Trans. Inf. Theory}, vol. 20, pp. 391-394, 1974.
\bibitem{fossorier} M. P. Fossorier and S. Lin, ``Soft-decision decoding of linear block codes based on ordered statistics," {\em IEEE Transactions on Information Theory}, vol. 41, no. 5, pp. 1379-1396, May 1995.
\bibitem{duffy-tit} K. R. Duffy, J. Li, and M. M\'edard, ``Capacity-achieving guessing random additive noise decoding," {\em IEEE Transactions on Information Theory}, vol. 65, no. 7, pp. 4023-4040, July 2019.
\bibitem{duffy-sgrand} K. R. Duffy and M. M\'edard, ``Guessing random additive noise decoding with soft detection symbol reliability information-SGRAND," in {\em IEEE International Symposium on Information Theory (ISIT)}, Paris, France, July 2019.
\bibitem{duffy-orbgrand} K. R. Duffy, ``Ordered reliability bits guessing random additive noise decoding," in {\em IEEE International Conference on Acoustics, Speech and Signal Processing (ICASSP)}, Toronto, Canada, June 2021.
\bibitem{abbas-orbgrand} S. M. Abbas, T. Tonnellier, F. Ercan, M. Jalaleddine, and W. J. Gross, ``High-throughput VLSI architecture for soft-decision decoding with ORBGRAND," in {\em IEEE International Conference on Acoustics, Speech and Signal Processing (ICASSP)}, Toronto, Canada, June 2021.
\bibitem{riaz} A. Riaz et al., ``Multi-Code Multi-Rate Universal Maximum Likelihood Decoder using GRAND," ESSCIRC 2021 - IEEE 47th European Solid State Circuits Conference (ESSCIRC), 2021, pp. 239-246.
\bibitem{an-grand-mo}W. An, M. M\'edard and K. R. Duffy, ``Keep the bursts and ditch the interleavers," {\em GLOBECOM 2020 - 2020 IEEE Global Communications Conference}, 2020, pp. 1-6
\bibitem{condo-grand} C. Condo, V. Bioglio and I. Land, ``High-performance low-complexity error pattern generation for ORBGRAND decoding," 2021 {\em IEEE Globecom Workshops}, 2021, pp. 1-6.
\bibitem{arikan2} E. Ar\i kan, ``From sequential decoding to channel polarization and back again," arXiv preprint arXiv:1908.09594 (2019).
\bibitem{rowshan-pac1} M. Rowshan, A. Burg and E. Viterbo, ``Polarization-adjusted Convolutional (PAC) Codes: Fano Decoding vs List Decoding," in \emph{IEEE Transactions on Vehicular Technology}, vol. 70, no. 2, pp. 1434-1447, Feb. 2021.
\bibitem{rowshan-precoding} M. Rowshan and E. Viterbo, ``On Convolutional Precoding in PAC Codes," 2021 {\em IEEE Globecom Workshops (GC Wkshps)}, Madrid, Spain, 2021, pp. 1-6.
\bibitem{rowshan-err_coef} M. Rowshan, S.H. Dau, and E. Viterbo, ``Error Coefficient-reduced Polar/PAC Codes," arXiv preprint  arXiv:2111.088435 (2021).
\bibitem{code} https://github.com/mohammad-rowshan/Constrained-Error-Pattern-Generation-for-GRAND

\end{thebibliography}
\end{document}